\documentclass[11pt]{article}
\usepackage{amsfonts}
\usepackage{amsthm}
\usepackage{amsmath}
\usepackage{graphicx}
\usepackage[algoruled,linesnumbered]{algorithm2e}

\newcommand{\nuple}[2]{{#1}_1,{#1}_2,\dots,{#1}_{#2}}
\newcommand{\parang}[1]{\langle {#1} \rangle}
\newcommand{\database}{\parang{U, {\cal R}}}
\newcommand{\AutR}{{\rm Aut}({\cal R})}
\newcommand{\AutRS}{{\rm Aut}({\cal R} \cup \{S\})}
\newcommand{\cgR}{{\rm cgr}({\cal R})}

\newcommand{\BIR}{{\rm BI}({\cal R})}
\newcommand{\BIRS}{{\rm BI}({\cal R} \cup \{S\})}

\newcommand{\stronguniform}[3]{{#2} \hookrightarrow_{#3} {#1}}

\newcommand{\CPR}{{\rm CP}({\cal R})}

\newcommand{\CPRS}{{\rm CP}({\cal R} \cup \{S\})}

\newcommand{\Imm}{Im}
\newcommand{\Imma}{\Imm}
\newcommand{\names}{names}
\newcommand{\name}{name}
\newcommand{\OPR}{{\rm OP}({\cal R})}
\newcommand{\OPRS}{{\rm OP}({\cal R} \cup \{S\})}
\newcommand{\Pow}{Pow}
\newcommand{\AM}{PD}

\newcommand{\dom}{Dom}
\newcommand{\Dom}{Domain}
\newcommand{\Inst}{Ext}
\newcommand{\BI}{BI}

\newcommand{\g}{g}
\newcommand{\f}{f}
\newcommand{\I}{\mathcal{I}}

\newcommand{\struct}{\ensuremath{\mathcal S}}
\newcommand{\PBI}[1]{P^{BI}_{#1}}

\begin{document}

\title{A Unifying Framework to Characterize the  Power of a Language to Express
  Relations\footnote{This work has been
    supported by the Italian Electrical Energy Company under research contracts
    ENEL/CRA n. R23TC0012 and R23WC0012.}}
\author{%
  Paola Bonizzoni\footnote{Dipartimento di Informatica,
    Sistemistica e Comunicazione, Universit{\`a} degli Studi di Milano -
    Bicocca, via Bicocca degli Arcimboldi 8, 20126 Milano, Italy,
    e-mail:  \{bonizzoni, leporati, mauri\}@disco.unimib.it}
  \and
  Peter J. Cameron\footnote{School of Mathematical Sciences,
    Queen Mary and Westfield College, London E1 4NS, UK, e-mail:
    p.j.cameron@qmw.ac.uk}
  \and
  Gianluca Della Vedova\footnote{Dipartimento di Statistica,
    Universit{\`a} degli Studi di Milano -
    Bicocca, via Bicocca degli Arcimboldi 8, 20126 Milano, Italy,
    e-mail:  gianluca.dellavedova@unimib.it }
  \and
  Alberto Leporati$^\dag$
  \and
  Giancarlo Mauri$^\dag$}
\date{}
\maketitle

\theoremstyle{plain}
\newtheorem*{theorem2}{General Theorem}
\newtheorem{PB}{Problem}[section]
\newtheorem{Cl}{Claim}[section]
\newtheorem{Def}{Definition}[section]
\newtheorem{Exercise}{Exercise}[section]
\newtheorem{Program}{Program}[section]
\newtheorem{Thm}{Theorem}[section]
\newtheorem{Rem}{Remark}[section]
\newtheorem{Lem}{Lemma}[section]
\newtheorem{Obs}{Observation}[section]
\newtheorem{Prop}[Thm]{Proposition}
\newtheorem{Cor}[Thm]{Corollary}
\newtheorem{Example}{Example}[section]

\begin{abstract}
  In this extended abstract we  provide a unifying framework that can be used to
  characterize and compare
  the expressive power of query languages for different data base models.
  The framework is based upon the new idea of {\em valid} partition,
  that is a partition of the elements of a
  given data base, where each class of the partition is composed by elements
  that cannot be separated (distinguished) according to some level of information
  contained in the data base.
  We describe two applications of this new framework, first by deriving
  a new syntactic characterization  of the expressive power
  of relational algebra which is equivalent to the one given by Paredaens,
  and subsequently by studying the expressive power of a simple graph-based
  data model.
\end{abstract}

\section{Introduction}

The relational data base model, introduced by Codd in~\cite{Codd1},
has been particularly successful since it is a mathematically
elegant model well suited to describe almost all ``real world''
situations. Since the query languages associated to such model (the
\emph{relational algebra} and  the \emph{relational calculus}) have a formal and
simple definition,
an interesting field of research is to study the expressive
power of such language. Codd~\cite{Codd} has proved that the relational
algebra is equivalent to the \emph{relational calculus}, in the sense that
both query languages can compute the same set of relations.

A breakthrough in this field~\cite{Banc1,Par78} has been a syntactic
characterization of the set of relations that can be computed in a give data
base. These results, also known as \emph{BP-completeness}, are based on the
principle of data independency from the physical representation: the
information that can be extracted from the data base is completely determined
at the logical level of such data base. This fact  can be stated in a simple
way: a relation $R$ can be computed from a data base $D$ if and only if all
permutations over the elements of $D$ which preserve $D$ (that is, all
permutations that produce a data base isomorphic to $D$), also preserve $R$.
An interesting interpretation of this property is that only the information
given by the structure of the data can be used to differentiate data values;
consequently, a query is expressible if and only if it does not add any
additional differentiation to the one initially available~\cite{ABGV}.

This idea can be rephrased by stating that the result of a query is invariant
w.r.t. permutations of indistinguishable values; such a permutation was
captured with the notion of automorphism in~\cite{Banc1,Par78}.
While the $BP$-criterion is a natural requirement, it refers to properties
of relations in a given data base instead of queries as a whole. We recall that
a query is an expression of the query language  that can be applied to different
data bases leading to possibly different results.
Thus it has been extended to a property of queries as partial functions from
data bases to data bases, which is known nowadays as {\em genericity}
\cite{CH80}: it has been recognized as the capability of the calculus to
preserve isomorphisms between data bases, rather than automorphisms.
Genericity is a common requirement for query languages and it is traditionally
related to the \emph{data independence} principle that assumes that the
data base is constructed over an abstract domain which is independent from
the internal representation of data.
Subsequent research has shown that this approach to the analysis of the
expressiveness of a query language has certain shortcomings
\cite{ABGV,HH96}, mainly when new data models, such as the object-based model,
are introduced.
Other notions have been proposed to analyze properties of queries in some new
models~\cite{BMS96,VVAG92,jacm/AbiteboulK98} pointing out the importance of extending
genericity to be used in more complex models.
In~\cite{BMS96} languages are classified  w.r.t. the degree of the use of the
equality predicate, by analyzing the invariance property of queries under
different mappings (not necessarily isomorphisms) over the data domain, which
are compatible with the relational structure of the data base.

Subsequent advances in data base theory have led to different models that
take into account the limitations of the relational model when it comes to
describe complex situations.
Most of such models have been introduced in the graph-based or object-oriented
frameworks, but usually their mathematical foundations do not allow a complete
study of the expressive power of the query languages introduced.
In fact, to our knowledge, the only exception is the graph-based model GOOD
\cite{AP91}.


In this paper we introduce a different syntactic characterization of queries
computable in a data base.
Our characterization relies upon the notion of partitions of the domain, where
each partition represents a level of {\em undifferentiation} among objects,
values or vertices. Notice that an automorphism also can represent a certain
level of undifferentiation.
Initially we will exploit such notion to give two new characterizations
of relations expressible in a relational data base. Subsequently, we will show
how to apply the new framework to analyze a simple graph-based model, hence
proving that our characterization can be useful in comparing the expressive
power of different data languages.

Following the approach of~\cite{Par78}, the data models studied in this paper
are domain-preserving, that is, it is not
possible to create new vertices or values, but only to query an existing data
base.
In our framework, a binary relation over sets of data values is defined,
denoted by $\hookrightarrow$, which relates those sets of values that cannot
be differentiated.
From the relation $\hookrightarrow$ we build some sets of partitions that
\emph{respect} $\hookrightarrow$, that is, all classes in a partition are
preserved by $\hookrightarrow$.
We prove that expressiveness of a query language can be stated as the
conservation of some of those partitions, where the exact set of partitions that
must be preserved depends on the data model.
The expressibility results we obtain have the following form: \emph{Given a
  data base $D$, let $S$ be a relation or a graph over the domain set of $D$.
  Then $S$ can be expressed in $D$ if and only if ${\rm P}({D})={\rm P}({D}
  \cup \{ S \})$}, where ${\rm P}({D})$ and ${\rm P}({D} \cup \{ S \})$ are two
sets of partitions which depend on the model under consideration.

\section{Preliminaries}

All sets considered in this paper are assumed to be finite and nonempty.
Given a set $U$, a \emph{relation} $R$ over $U$ is a subset of the cartesian
product $U^a=U \times U \times \cdots \times U$ ($a$ times) for some fixed
integer $a > 0$, that is a set of tuples of length $a$,
where all components of a tuple are elements of $U$.
The  number $a\in\mathbb{N}$ is called the \emph{order} or \emph{arity} of the
relation.
Given a set ${\cal R} = \{\nuple{R}{p}\}$ of relations over $U$, the pair
$\database$ is called a \emph{relational database}; in this setting, $U$ is
the \emph{domain} of the database, and ${\cal R}$ is the set of relations of
the database.

Given a relation $R \in {\cal R}$ of a database $\database$, we denote with
$D(R)$ the \emph{data domain} of $R$, that is the subset of the elements of
the database domain $U$ that are in at least one tuple of $R$.
The notion of data domain is easily extended to the set ${\cal R}$ of relations
as the set union of relations' data domains: $D({\cal R}) = \bigcup_{R \in
  {\cal R}}D(R)$.
Without loss of generality, we can assume that $D({\cal R}) = U$ for every
considered database  $\database$.
This seemingly trivial requirement is indeed very important, as it will become
evident after Theorem~\ref{teo:Paredaens2}, therefore we will omit the universe
set unless it is necessary to avoid any ambiguities.

Just as in~\cite{Par78}, when referring to a relational database, we use the
\emph{relational algebra}
as a query language.
In relational algebra two binary operators (union and product) and three
unary operators (projection, equality restriction and inequality restriction)
are given.
In the following definition all relations are
defined over the same database domain $U$.

\begin{Def}[Relational Algebra]
  \label{def:relational-algebra}
  Let $R$ and $S$ be two relations with the same arity; the \emph{union} of $R$
  and $S$, denoted by $R \cup S$, is simply the set--theoretical union of the
  two sets of tuples.

  Given two relations $R$ and $S$ (not necessarily with the same arity), the
  (cartesian) \emph{product} of $R$ and $S$, denoted by $R \times S$, is the set
  of all possible concatenations of a tuple of $R$ with a tuple of $S$:
  $   \{r \cdot s \, | \, r \in R, \, s \in S\}$.
  The abbreviation $R^k$ is used to express the relation $R \times \cdots \times
  R$ ($k$ times).

  Let $m$ be the arity of a relation $R$, $q \le m$ a positive integer and
  $f : \{1,$ $\ldots, q\} \to \{1, \ldots, m\}$ a function.
  The \emph{projection} of $R$ over $(f(1), \ldots, f(q))$, denoted by
  $R\pi\big(f(1), \ldots, f(q)\big)$, is the relation:
  $\big\{(r_{f(1)},\ldots, r_{f(q)}) \, : \, (r_1, \ldots, r_m) \in R)\big\}$.

  Now, let $j_1$ and $j_2$ be two integers such that $1 \leq j_1, j_2 \leq m$,
  where $m$ is the arity of a relation $R$.
  The \emph{equality restriction} of $R$ on $j_1$ and $j_2$ is the relation,
  denoted by $R \, | \, j_1 = j_2$, that is obtained by taking from $R$ all the
  tuples for which the $j_1$-th and the $j_2$-th components are equal:
  $\big\{(r_1, \ldots, r_m) \in R \, : \, r_{j_1} = r_{j_2}\big\}$.
  Analogously, the \emph{inequality restriction} of $R$ on $j_1$ and $j_2$,
  denoted by $R \, | \, j_1 \neq j_2$, is the relation obtained by taking
  from $R$ all the tuples for which the $j_1$-th and the $j_2$-th components are
  different: $\big\{(r_1, \ldots, r_m) \in R \, : \, r_{j_1} \neq r_{j_2}\big\}$.
\end{Def}

The five operations just described are sufficient to generate the operations
of intersection, difference, join and division, usually assumed as primitives
in Codd's relational algebra; a proof of this fact can be found, for example,
in~\cite{Codd}.

Given a relational data base $D=\database$, we will denote by $M_E(D)$ the
relation which is the result of applying the expression (of the relational
algebra) E to the data base $D$. Moreover a relation $S$ over $U$ is told to
be \emph{expressible} from ${\cal R}$ if there exists an expression $E$ whose
operands are all relations in ${\cal R}$, and such that $M_E(D)$ is equal to
$S$.
Following~\cite{Par78}, we denote with $\BIR$ (basic information
contained in the set of relations ${\cal R}$) the set of relations that can
be expressed from ${\cal R}$.

As  observed in~\cite{Par78}, $\BIR$ is the set of the answers to all
possible queries that can
be asked  to
a relational datMabase that contains the relations ${\cal R}$.
In~\cite{Par78}, Paredaens gives a characterization of the class $\BIR$
based upon appropriate automorphisms, that is permutations of the elements of
the database domain.

Let $R$ be a relation of order $m$ over a set $U$.
As in~\cite{Par78}, an \emph{automorphism} is a bijective function (that
is, a permutation) on $U$.
We say that the automorphism $\psi: U \to U$ \emph{respects} the relation $R$
or, equivalently, that $\psi$ is $R$-\emph{compatible} if, for each tuple
$(\nuple{a}{m}) \in U^m$, $(\nuple{a}{m}) \in R \Longrightarrow
(\psi(a_1), \psi(a_2), \ldots, \psi(a_m)) \in R$.

The compatibility of an automorphism $\psi: U \to U$ with respect to a
relation $R$ can be naturally extended to a set ${\cal R}$ of relations in
the following way: $\psi$ \emph{respects} the relations in ${\cal R}$ or,
equivalently, $\psi$ is ${\cal R}$-\emph{compatible} if $\psi$ is
$R$-compatible for each relation $R$ in ${\cal R}$. Notice that the set of
automorphisms ${\cal R}$-\emph{compatible}, is a group\footnote{a group consists
  of a set $G$ of elements, a binary associative operation  on $G$, and an identity
  element $1_G\in G$, such that the operation is closed and invertible in $G$}
where the operation is
the composition of functions and the identity is the identity function (i.e. the
function defined as $f(x)=x$).
As in~\cite{Par78}, we denote with $\AutR$ the set of all the
automorphisms $\psi: U \to U$ which are ${\cal R}$-compatible; with a small
abuse of notation, if ${\cal R} = \{ R \}$, we will usually write
${\rm Aut}(R)$ instead of ${\rm Aut}(\{R\})$.
It will be very useful to consider the following representation of $\AutR$.

\begin{Def}
  Let $\database$ be a relational database, with $U = \{ d_1, d_2, \ldots$,
  $d_n \}$, and
  let $\AutR = \{\nuple{\psi}{l}\}$ be the set of ${\cal R}$-compatible
  automorphisms.
  The following relation of arity $n$:
  \begin{equation*}
    \cgR = \begin{tabular}{ccc}
      $\psi_1(d_1)$ & $\cdots$ & $\psi_1(d_n)$ \\
      $\vdots$ & $\ddots$ & $\vdots$ \\
      $\psi_l(d_1)$ & $\cdots$ & $\psi_l(d_n)$
    \end{tabular}
  \end{equation*}
  is called the \emph{cogroup--relation} of $\database$.
\end{Def}

As we can see, each row (tuple) of the relation $\cgR$ represents one of the
${\cal R}$-compatible automorphisms.
Since we do not associate any particular meaning to the elements of the domain
$U$, if $|U| = n$ we can assume, without loss of generality, $U = \{ 1, 2,
\ldots, n \}$.
We can also assume that the first tuple of $\cgR$ represents the identity
function on $U$ (which is always present in $\AutR$, since it is compatible
with every nonempty set of relations); as a consequence, it can always be
assumed that the first row of $\cgR$ is the tuple $(1, 2, \ldots, n)$.

\begin{Example}
  Let $\database$ be a relational data base, with:
  \begin{itemize}
  \item $U = \{1, 2, 3, 4\}$
  \item ${\cal R} = \{R_1, R_2, R_3\}$, with:
    \begin{displaymath}
      R_1 = \, \begin{tabular}{cc}
        $1$ & $2$ \\
        $2$ & $1$ \\
        $3$ & $4$ \\
        $4$ & $3$
      \end{tabular}
      \qquad
      R_2 = \, \begin{tabular}{cc}
        $1$ & $3$ \\
        $3$ & $1$ \\
        $2$ & $4$ \\
        $4$ & $2$
      \end{tabular}
      \qquad
      R_3 = \, \begin{tabular}{cc}
        $1$ & $4$ \\
        $4$ & $1$ \\
        $2$ & $3$ \\
        $3$ & $2$
      \end{tabular}
    \end{displaymath}
  \end{itemize}
  It is easily verified that:
  \begin{displaymath}
    {\rm Aut}(\{R_1, R_2\}) = {\rm Aut}(\{R_1, R_3\}) = {\rm Aut}(\{R_2, R_3\})
    = {\rm Aut}(\{R_1, R_2, R_3\})
  \end{displaymath}
  \begin{displaymath}
    \cgR = \begin{tabular}{cccc}
      $1$ & $2$ & $3$ & $4$ \\
      $2$ & $1$ & $4$ & $3$ \\
      $3$ & $4$ & $1$ & $2$ \\
      $4$ & $3$ & $2$ & $1$
    \end{tabular}
  \end{displaymath}
  If we look at the ${\cal R}$-compatible automorphisms as permutations over
  $U$, we can express $\AutR$ as follows:
  \begin{displaymath}
    \AutR = \begin{tabular}{l}
      {\rm Identity} \\
      {\rm ($1$ $2$) ($3$ $4$)} \\
      {\rm ($1$ $3$) ($2$ $4$)} \\
      {\rm ($1$ $4$) ($2$ $3$)}
    \end{tabular}
  \end{displaymath}
  \label{ex:Klein-group}
\end{Example}

It is not difficult to see that, for a given database $\database$, the set
$\AutR$ of ${\cal R}$-compatible automorphisms is indeed a group with respect
to function composition, with the identity function over $U$ as unitary
element.
In fact, the identity over $U$ is always in $\AutR$, the inverse of an
${\cal R}$-compatible automorphism is still an ${\cal R}$-compatible
automorphism, and the composition between two ${\cal R}$-compatible
automorphisms is again an ${\cal R}$-compatible automorphism.
Since we can always assume $U = \{ 1, 2, \ldots, n \}$, we can think of
$\AutR$ as a finite permutation group over the set $\{ 1, 2, \ldots, n \}$,
that is a subgroup of the symmetric group $S_n$.

In this paper we investigate the relation between expressive power and
partitions of the database domain.
More precisely, we investigate the possibility to characterize the expressive
power of relational and graph-based databases via one or more theorems abiding
to the following meta theorem.
\begin{Thm}[Meta theorem]
  \label{thm:scheme1}
  Let $\database$ be a relational database, and let $S$ be a relation over
  $U$.
  Then
  $S \in \BIR \iff {\rm P}({\cal R}) = {\rm P}({\cal R} \cup \{ S \} )$,
  where ${\rm P}({\cal R})$ and ${\rm P}({\cal R} \cup \{ S \})$ are sets of
  partitions over $U$, built from the sets ${\cal R}$ and ${\cal R} \cup
  \{ S \}$ of relations respectively.
\end{Thm}

\section{Expressiveness in Relational Databases}

The relevance of the main result in~\cite{Par78} is that it is the first
syntactic characterization of the relations that can be obtained from a given
database $\database$ when the relational algebra is used as a query language.
More precisely, in~\cite{Par78} the following theorem is proved.

\begin{Thm}
  Let $\database$ be a relational database, and let $S$ be a relation over
  $U$.
  Then
  $S \in \BIR \iff \AutR \subseteq {\rm Aut}(S)$ and $D(S) \subseteq D({\cal R})$.
  \label{teo:Paredaens}
\end{Thm}

Basically, Paredaens has been able to point out the fundamental relation
between expressiveness in a database and the set of automorphisms in the
relational model.
Such result has been successively extended in~\cite{CH80} to define in a
formal way the notion of \emph{genericity}, that is computable queries
\cite{CH80} have to
be invariant with respect to the isomorphisms between databases.
We can restate Theorem~\ref{teo:Paredaens} in a form that will be more
convenient for our purposes.
\begin{Thm}
  Let $\database$ be a relational database, and let $S$ be a relation over
  $U$.
  Then $S \in \BIR \iff \AutR = \AutRS$.
  \label{teo:Paredaens2}
\end{Thm}
\begin{proof}
  First of all, we show that $S \in \BIR \iff \BIR = \BIRS$.
  Proving that $S \in \BIR \Longrightarrow \BIR = \BIRS$
  is trivial as $\BIR \subseteq \BIRS$. The latter stems from the fact that the
  relations which are expressible from ${\cal R}$ are those
  obtained from ${\cal R} \cup \{S\}$ simply ignoring the relation $S$.
  Let now be $S \in \BIR$ and $T \in \BIRS$.
  If the expression that gives $T$ from ${\cal R} \cup \{S\}$ does contain
  some occurrence of the relation $S$, it is sufficient to replace such
  occurrence with the expression that gives $S$ from ${\cal R}$ to conclude that
  $T \in \BIR$, and thus $\BIRS \subseteq \BIR$.
  It is immediate to notice that $\BIR = \BIRS$ implies $S \in \BIRS$.

  Since we have  established that  $S \in \BIR \iff \BIR = \BIRS$,  the two
  databases $\database$ and $\parang{U, {\cal R} \cup \{S\}}$ are \emph{basic
    information equivalent} -- that is, every relation of the first database
  can be obtained from the relations of the second database and vice versa --
  if and only if $S$ is expressible from $\database$.
  A direct consequence of Theorem~\ref{teo:Paredaens} is that two databases
  $\parang{U, {\cal R}_1}$ and $\parang{U, {\cal R}_2}$ are basic information
  equivalent if and only if $D({\cal R}_1) = D({\cal R}_2)$ (which are assumed
  to be both equal to $U$) and ${\rm Aut}({\cal R}_1) = {\rm Aut}({\cal R}_2)$;
  thus, we can conclude that $S \in \BIR \iff \AutR = \AutRS$
  as stated.
\end{proof}

We observe that, given our assumption that $U = D({\cal R})$, in Theorem
\ref{teo:Paredaens2} we can get rid of the inclusion between the domains, since
it is implicit from the fact that $S$ is a relation over $U$.
On the other hand, we {\em cannot} ignore the inclusion condition if we
suppose that $D({\cal R}) \subset U$, since in such a situation it is not
difficult to show two relations $R$ and $S$ such that ${\rm Aut}(R) =
{\rm Aut}(\{R, S\})$ but $S \not\in {\rm BI}(R)$.


A notion that seems tightly related to the expressiveness of relations in a
database is that of \emph{indistinguishability} between elements of the
domain.
Intuitively, the idea is that the elements of a subset of the domain of a
given database are indistinguishable if and only if no query to the database
is able to divide the set in two parts, one made of the elements that occur in
the relation resulting from the query and the other made of the elements that
do not occur in the relation.
In such a situation, we say that the set of indistinguishable elements cannot
be \emph{separated} by any of the queries that can be presented to the
database.
Thus, a relation resulting from a query to the database can only contain all
or none of the elements of a non-separable set.

Theorem~\ref{thm:scheme1} defines the general framework we propose to
investigate the expressive power of query languages.
In this framework different notions of expressible queries can be studied by
considering different sets of partitions.
For a given database $\database$, we say that a set  ${\rm P}({\cal R})$ of
partitions of $U$ is a set of \emph{valid partitions} if and
only if it satisfies Theorem~\ref{thm:scheme1}.
By the results  in~\cite{Par78},
it seems to us quite natural to define the following sets of valid partitions,
namely the \emph{orbit partitions} and the \emph{cycle partitions}; indeed
later we will be able to prove that, in the context of
Theorem~\ref{thm:scheme1},   they are equivalent to the characterization of
relations obtainable in a relational data base of~\cite{Par78}.

\begin{Def}
  Let $\database$ be a relational database, and let ${\cal P} = \{ P_1, P_2$,
  $\ldots, P_k \}$ be a partition of $U$.
  ${\cal P}$ is an \emph{orbit partition} of $U$ with respect to ${\cal R}$
  if both the following conditions hold:
  \begin{enumerate}
  \item for each relation $R \in {\cal R}$ and for each class $P_i \in
    {\cal P}$, $P_i \cap D(R) = \emptyset$ or $P_i \subseteq D(R)$;
  \item for each class $P_i \in {\cal P}$ and for each pair $a_1, a_2$
    of elements of $P_i$ there exists an automorphism $\phi \in \AutR$
    such that $\phi(a_1) = a_2$, and $\phi(P_j) = P_j$ for every
    class $P_j \in {\cal P}$.
  \end{enumerate}
  We denote with $\OPR$ the set of all orbit partitions of the given
  database $\database$.
  \label{def:orbit-partition}
\end{Def}

\begin{Def}
  Let $\database$ be a relational database, and let ${\cal P} = \{ P_1, P_2$,
  $\ldots, P_k \}$ be a partition of $U$.
  ${\cal P}$ is a \emph{cycle partition} of $U$ with respect to ${\cal R}$
  if both the following conditions hold:
  \begin{enumerate}
  \item for each relation $R \in {\cal R}$ and for each class $P_i \in
    {\cal P}$, $P_i \cap D(R) = \emptyset$ or $P_i \subseteq D(R)$;
  \item there exists an automorphism $\phi \in \AutR$ such that for each
    class $P_i \in {\cal P}$ and for each pair $a_1, a_2$ of elements
    of $P_i$ there exists an integer $n$ such that $\phi^n(a_1) = a_2$
    and $\phi(P_j) = P_j$ for every class $P_j \in {\cal P}$.
  \end{enumerate}
  We denote with $\CPR$ the set of all cycle partitions of the given
  database $\database$.
  \label{def:cycle-partition}
\end{Def}
As already stated for $\AutR$, if $R$ is a relation we will write ${\rm OP}
(R)$ and ${\rm CP}(R)$ instead of ${\rm OP}(\{R\})$ and ${\rm CP}(\{R\})$
respectively.

The following theorem is an alternative formulation of the main result
of~\cite{Par78} (the equivalence of the two formulation follows from Theorem
\ref{teo:Paredaens2}) which is more useful for our purposes.

\begin{Thm}
  Let $\database$ be a relational database, and let $S$ be a relation over
  $U$.
  Then
  $\AutR = \AutRS \iff {\rm P}({\cal R}) = {\rm P}({\cal R} \cup \{ S \} )$.
\end{Thm}


Let $\database$ be a relational database, and let $\AutR$ and $\cgR$ be
respectively the group of ${\cal R}$-compatible automorphisms and the
cogroup-relation of ${\cal R}$.
A useful fact proved in~\cite{Par78} is that the cogroup-relation is
expressible from ${\cal R}$, that is $\cgR \in \BIR$.
Using this fact, we are able to prove the following theorem.

\begin{Thm}
  $\AutR = {\rm Aut}(\cgR)$.
  \label{teo:AutR=cgR}
\end{Thm}
\begin{proof}
  Since  $\cgR \in \BIR$, by Theorem~\ref{teo:Paredaens} we can
  conclude that $\AutR \subseteq {\rm Aut}(\cgR)$.
  Now, let $\phi \in {\rm Aut}(\cgR)$; as we have already observed, $\phi$ is a
  permutation of the set $U = D({\cal R})$, as well as of the tuples that
  compose the relation $\cgR$.
  Thus, for each tuple $t \in \cgR$, we have that $\phi(t) \in \cgR$.
  In particular, by letting $n$ be the cardinality of $U$, we have:
  \begin{equation*}
    \phi\big((1, 2, \ldots, n)\big) = \big(\phi(1), \phi(2), \ldots, \phi(n)
    \big) \in \cgR
  \end{equation*}
  Thus, the elements of $U$ are mapped by $\phi$ in such a way that the result
  is a row of the cogroup-relation; so we can conclude that $\phi \in \AutR$.
\end{proof}

A direct consequence of Theorem~\ref{teo:AutR=cgR} is that not only $\cgR
\in \BIR$, as established by Paredaens, but also $R \in {\rm BI}(\cgR)$ for
every relation $R \in {\cal R}$, since $D(R) \subseteq D(\cgR) = U$ and
${\rm Aut}(\cgR) = \AutR \subseteq {\rm Aut}(R)$.
As a corollary of Theorem~\ref{teo:AutR=cgR}, if we are interested to study
the expressive power of a given relational database $\database$ then we can
work as well on the database $\parang{U, \{\cgR\}}$, which has only one
relation and, moreover, such relation is an explicit representation of the
finite permutation group $\AutR$.

We now turn our attention to the structure of $\OPR$ and $\CPR$.
First of all we observe that, thanks to Theorem~\ref{teo:AutR=cgR}, we can
get rid of item 1 in Definitions~\ref{def:orbit-partition} and
\ref{def:cycle-partition} since, by considering the database $\parang{U,
  \{ \cgR \}}$, there is only one relation and, for such relation, it holds
$P_i \subseteq D(\cgR) = U$ for each $P_i \in {\cal P}$.

To characterize the sets of cycle and orbit partitions we need to recall
some notions from basic abstract algebra.

\begin{Def}
  Let $X$ be a set and $\parang{G, \cdot, e}$ a group.
  An \emph{action of $G$ on $X$} is a map $\ast : G \times X \to X$ such that
  \begin{enumerate}
  \item $\forall \, x \in X, \quad e \ast x = x$;
  \item $\forall \, g_1, g_2 \in G, \; \forall \, x \in X \quad (g_1 \cdot
    g_2) \ast x = g_1 \ast (g_2 \ast x)$
  \end{enumerate}
\end{Def}

In group theory it is customary to omit the operators symbols from expressions
when confusion does not arise; so, the expression in item 2 above is usually
written as: $(g_1 g_2)x = g_1 (g_2 x)$.

\begin{Def}
  Let $G$ be a group acting on a set $X$.
  For $x_1, x_2 \in X$, let $x_1 \sim x_2$ if and only if there exists $g \in
  G$ such that $gx_1 = x_2$.
  It is not difficult to see that $\sim$ is an equivalence relation on $X$,
  and thus it induces a partition ${\cal P}$ on $X$.
  The classes of ${\cal P}$ are called the \emph{orbits} in $X$ under $G$.
  If $x \in X$, the class containing $x$ --- denoted by $Gx$ --- is called
  the \emph{orbit of} $x$ under $G$.
  In other words, $Gx = \{y \in X \: | \: y = gx \; {\rm for \; some} \; g
  \in G\}$.
\end{Def}

It is not difficult to see that the partition induced by the orbits of
$\AutR$ on $U$ satisfies Definition~\ref{def:orbit-partition}.
In fact, every automorphism $\phi \in \AutR$ maps each orbit $\AutR x$ into
itself and, given a  pair $a_1, a_2$ of elements of $U$, there exists an
automorphism that maps $a_1$ to $a_2$ if and only if $a_1$ and $a_2$ are in
the same orbit.
Moreover, if $H$ is a subgroup of a group $G$ acting on the set $X$, then
every orbit $Hx$ is a subset of the orbit $Gx$; more precisely, it is not
difficult to prove that the orbits induced by $H$ are a \emph{refinement}
of the orbits induced by $G$.
Since each partition induced by the orbits of every subgroup of $\AutR$
satisfies Definition~\ref{def:orbit-partition}, we have that $\OPR$ contains
the set of those partitions.

Vice versa, let ${\cal P} \in \OPR$.
It is not difficult to see that the set of automorphisms $\phi \in \AutR$ that
map each class of ${\cal P}$ into itself and that map each element of a class
to an element of the same class forms a subgroup of $\AutR$; moreover, the
orbit partition induced by such a subgroup is just ${\cal P}$.
As a consequence, $\OPR$ is a subset of the set of partitions induced by all
the subgroups of $\AutR$; since also the converse inclusion holds, the two
sets indeed coincide.

\begin{Def}
  Let $G$ be a group acting on the set $X$, and let $g \in G$.
  For $x_1, x_2 \in X$, let $x_1 \sim x_2$ if and only if there exists an
  integer $n$ such that $x_2 = g^nx_1$, where $g^n$ is the application of
  $g$ for $n$ times.
  It is not difficult to see that $\sim$ is an equivalence relation on $X$,
  and thus it induces a partition ${\cal P}$ on $X$.
  The classes of ${\cal P}$ are called the \emph{cycles} of $g$ on $X$.
\end{Def}

Analogously to what said about orbits, it is not difficult to see that the
partitions induced by the cycles of the automorphisms of $\AutR$ satisfy
Definition~\ref{def:cycle-partition}.
We observe that, while an orbit partition is induced by a subgroup of $\AutR$,
a cycle partition is induced by an automorphism, that is by an element of
$\AutR$.
The class $\CPR$ is thus the set of cycle partitions obtained by considering
every element of $\AutR$.

\begin{Def}
  Let $G$ be a group acting on the set $X$ and let $g$ be a permutation in
  $G$.
  Then the \emph{orbits of the}
  (cyclic) \emph{group} $\parang{g}$ generated by $g$ are the cycles of $g$.
  Since $\parang{g}$ is a subgroup of $G$, we have immediately that every
  cycle partition of $G$ is also an orbit partition of $G$, that is,
  $\CPR \subseteq \OPR$.
\end{Def}

Example~\ref{ex:Klein-group} can be used to show that the converse does not
generally hold: not every orbit partition is also a cycle partition.
In fact we have:
\begin{align*}
  \CPR &= \Big\{\big\{\{1\},\{2\},\{3\},\{4\}\big\},
  \big\{\{1,2\},\{3,4\}\big\}, \\
  & \hspace{0.8cm} \big\{\{1,3\},\{2,4\}\big\},
  \big\{\{1,4\},\{2,3\}\big\}\Big\} \\
  \OPR &= \CPR \cup \Big\{\big\{\{1,2,3,4\}\big\}\Big\}
\end{align*}

As noted above, Theorem~\ref{teo:AutR=cgR} allows us to deal only with
cogroup-relations instead of sets of arbitrary relations.
The same can be done when working with cycle and orbit partitions: since
cycles and orbits that form the partitions in $\CPR$ and $\OPR$ are
completely determined from the elements and the subgroups of $\AutR$
respectively, by Theorem~\ref{teo:AutR=cgR} we can conclude that
$\CPR = {\rm CP}(\cgR)$ and $\OPR = {\rm OP}(\cgR)$.

It is possible to show that both the set $\CPR$ of cycle partitions and the
set $\OPR$ of orbit partitions of a given database $\database$ constitute a
partially ordered set (poset) with respect to the binary relation $\le$, where
${\cal P}_1 \le {\cal P}_2$ iff each class of ${\cal P}_1$ is
contained in some class of ${\cal P}_2$, where ${\cal P}_1$ and ${\cal P}_2$
are two partitions in ${\rm P}({\cal R})$,
${\rm P}({\cal R})$ is equal to $\CPR$ or $\OPR$.
In fact, it is not difficult to see that $\le$ is reflexive, antisymmetric and
transitive: that is, $\le$ is an order relation over both $\CPR$ and $\OPR$.
One notably difference between the posets $\parang{\OPR, \le}$ and
$\parang{\CPR, \le}$ is that the first has always a maximum element,
corresponding to the orbits of the entire $\AutR$, while the second may not
have a maximum element, as shown above referring to Example
\ref{ex:Klein-group}, where $\AutR$ is the so called Klein group.
Instead, both the posets have a minimum element, corresponding to the cycles
(equal to the orbits) induced by the identity element of $\AutR$: that is,
the trivial partition, where each class is a singleton.

In order to prove our main results we need some definitions and some well
known properties of finite groups.
Here we just recall the notion of \emph{stabilizer}; we address the reader to
an introductory book on abstract algebra, such as
\cite{FirstCourseAbstractAlgebra}, for the
notion of coset and its properties.

\begin{Def}
  Let $G$ be a group acting on a set $X$, and let $x \in X$.
  The subgroup $G_x$ of $G$ defined as $G_x = \{g \in G \: | \: gx = x\}$
  is called the \emph{stabilizer} of $x$ in $G$.
\end{Def}
It is not difficult to see that if $G$ is a group which acts on the set $X$,
and $x \in X$, then the stabilizer $G_x$ of $x$ can be considered as a group
which acts on the set $X \setminus \{x\}$.
The following are two well known results in group theory: Lagrange's theorem
-- which correlates the cardinality of a given group $G$ and the cardinality
of a given subgroup $H$ of $G$ with the number of left cosets of $G$ with
respect to $H$ -- and a theorem which expresses the cardinality of the
orbit of $G$ containing $x$ as the number of left cosets of $G$ with respect
to the stabilizer $G_x$.

\begin{Thm}[Lagrange's Theorem]
  Let $G$ be a finite group, and let $H$ be a subgroup of $G$.
  Then $|G| = (G:H) \cdot |H|$,
  where $(G:H)$ is the number of left cosets of $G$ with respect to $H$,
  and is usually called the \emph{index of $H$ in $G$}.
\end{Thm}

\begin{Thm}
  Let $G$ be a finite group acting on a set $X$, and let $x \in X$.
  Then $|Gx| = (G:G_x)$, that is there exists a one-to-one correspondence between the
  elements of the orbit $Gx$ of $x$ under $G$ and the left cosets of the
  stabilizer $G_x$ in $G$.
  \label{teo:card-orbit}
\end{Thm}

We are now able to prove the following theorem.
\begin{Thm}
  Let $G$ be a subgroup of the symmetric group $S_n$, and let $H$ be a
  subgroup of $G$.
  If the orbit partitions of $G$ and $H$ are the same, then $H = G$.
  \label{teo:orbit-partitions}
\end{Thm}
\begin{proof}
  We prove the assertion by induction on $n$.
  For $n\le 2$ the theorem can be proved by direct
  inspection of the subgroups of $S_n$.

  Now, let us suppose that the theorem is true for $n-1$, and let us show that
  it holds also for $n$.
  We first observe that since the orbit partitions of $G$ and $H$ are the same,
  then also the orbits $Gn$ and $Hn$ of the element $n$ with respect to $G$ and
  $H$ are the same.
  Now, if we take all the partitions having $\{n\}$ as a class, we get the orbit
  partitions induced by the stabilizers $G_n$ and $H_n$ of the element $n$ with
  respect to $G$ and $H$.
  These orbit partitions are equal and thus, by induction hypothesis, $G_n =
  H_n$.
  By Lagrange's theorem, we can express the cardinalities of $G$ and $H$ with
  respect to the cardinalities of their stabilizers as $|G| = (G : G_n) \cdot
  |G_n|$ and $|H| = (H : H_n) \cdot |H_n|$.
  where $(G : G_n)$ and $(H : H_n)$ are the indices, respectively, of the
  stabilizer $G_n$ in $G$ and of the stabilizer $H_n$ in $H$.
  By Theorem~\ref{teo:card-orbit}, we can infer that $|G| = |Gn| \cdot |G_n|$
  and
  $|H| = |Hn| \cdot |H_n|$.
  Since $|Gn| = |Hn|$ and $|G_n| = |H_n|$, we can conclude that $G$ and $H$ have
  the same order, and thus $G = H$.
\end{proof}

Theorem~\ref{teo:orbit-partitions} allows us to show that the orbit partitions
of a given database satisfy Theorem Scheme II; in fact, the following theorem
provides a first characterization of expressible queries in relational
databases alternative to the one originally given by Paredaens.

\begin{Cor}
  Let $\database$ be a relational database, and let $S$ be a relation over
  $U$.
  Then       $\AutR = \AutRS \iff \OPR = \OPRS$
  \label{teo:scheme-with-orbits}
\end{Cor}
\begin{proof}
  If $\AutR = \AutRS$, since the orbit partitions are completely determined from the
  subgroups of $\AutR$, we obtain that $\OPR = \OPRS$.

  For the converse, we observe that $\AutR$ is a subgroup of the symmetric
  group $S_n$, and $\AutRS$ is a subgroup of $\AutR$.
  By hypothesis, the orbit partitions of $\AutR$ and $\AutRS$ are equal and
  thus, by Theorem~\ref{teo:orbit-partitions}, $\AutR = \AutRS$.
\end{proof}

A second characterization of expressible queries  in relational databases can
be obtained by considering cycle partitions instead of orbit partitions.
We need the following lemma.

\begin{Lem}
  Let $G$ be a subgroup of the symmetric group $S_n$, and let $H$ be a
  subgroup of $G$.
  If the cycle partitions of $G$ and $H$ are the same, then also the orbits
  of $G$ and $H$ are the same, that is $Hx = Gx$ for every $x \in \{ 1, 2,
  \ldots, n \}$.
  \label{lem:orbits-from-cycles}
\end{Lem}
\begin{proof}
  Since the orbit $Gx$ is the set of elements of $\{ 1, 2, \ldots, n \}$ which
  are reachable from $x$ through some element $g$ of $G$, while a cycle
  containing $x$ is the set of elements which are reachable from $x$ through
  one element $g$ of $G$, one method to build $Gx$ from the cycle partitions
  of $G$ is given by Algorithm~\ref{alg:1}.

  \begin{algorithm}[ht!]
    \caption{BuildOrbit}
    \label{alg:1}
    \KwData{an integer $x \in \{1, \ldots ,n\}$, a subgroup $G$ of $S_n$, a set
      $CP$ of cycle partitions}
    \KwResult{Result}
    Result $\gets$ $\{x\}$\;
    \Repeat{Modified = false}{%
      Modified $\gets$ false\;
      \ForEach{partition ${\cal P}$ in $CP$}{%
        Cycles $\gets$ the smallest union of cycles of ${\cal P}$ which covers
        Result\;
        \If{Cycles $\setminus$ Result $\neq \emptyset$}{%
          Result $\gets$ Result $\cup$ Cycles\;
          Modified $\gets$ true\;
        }
      }
    }
  \end{algorithm}

  Algorithm~\ref{alg:1} computes the least subset $O$ of $\{ 1, 2, \ldots, n \}$ which
  contains $x$ and such that, for every cycle partition ${\cal P}$ of $G$, $O$
  is the union of some cycles in ${\cal P}$; it is not difficult to see that
  $O$ is, indeed, the orbit $Gx$.

  Since the cycle partitions of $G$ and $H$ are the same by hypothesis, the
  orbits computed by the algorithm above will be the same for $G$ and $H$,
  for every choice of $x \in \{ 1, 2, \ldots, n \}$.
\end{proof}

We are now ready to prove the following theorem.

\begin{Thm}
  Let $G$ be a subgroup of the symmetric group $S_n$, and let $H$ be a
  subgroup of $G$.
  If the cycle partitions of $G$ and $H$ are the same, then $H = G$.
  \label{teo:cycle-partitions}
\end{Thm}
\begin{proof}
  By Lemma~\ref{lem:orbits-from-cycles}, the orbits of $G$ and $H$ are the
  same.
  Thus we can prove the theorem by the same argument used for Theorem
  \ref{teo:orbit-partitions}.
\end{proof}

A direct consequence of Theorem~\ref{teo:cycle-partitions} is that the cycle
partitions of a given database satisfy Theorem Schema II; thus, the following
theorem provides a second characterization of expressible queries in
relational databases alternative to the one originally given by Paredaens.
The proof is analogous to the one given for Theorem
\ref{teo:scheme-with-orbits}.

\begin{Thm}
  Let $\database$ be a relational database, and let $S$ be a relation over
  $U$.
  Then:
  \begin{equation*}
    \AutR = \AutRS \iff \CPR = \CPRS
  \end{equation*}
  \label{teo:scheme-with-cycles}
\end{Thm}

A final observation is due about Theorems~\ref{teo:scheme-with-orbits} and
\ref{teo:scheme-with-cycles}. Even though there
is a strong resemblance between our meta Theorem~\ref{thm:scheme1} and
Theorem~\ref{teo:Paredaens2}, our results \emph{cannot} be expressed neither in the
form $S \in \BIR \iff \OPR \subseteq {\rm OP}(S)$ and $D(S) \subseteq D({\cal
  R})$ nor in the form $S \in \BIR \iff \CPR \subseteq {\rm CP}(S)$  and
$D(S) \subseteq D({\cal R})$,
as shown in the next example.

\begin{Example}
  Let $\parang{U, \{ R \}}$ and $\parang{U, \{ S \}}$ be two relational
  databases, with:

  \begin{itemize}
  \item $U = \{1, 2, 3, 4, 5\}$
  \item R = \begin{tabular}{ccccc}
      $1$ & $2$ & $3$ & $4$ & $5$ \\
      $2$ & $3$ & $4$ & $5$ & $1$ \\
      $3$ & $4$ & $5$ & $1$ & $2$ \\
      $4$ & $5$ & $1$ & $2$ & $3$ \\
      $5$ & $1$ & $2$ & $3$ & $4$
    \end{tabular}
    \qquad
    S = \begin{tabular}{ccccc}
      $1$ & $2$ & $3$ & $4$ & $5$ \\
      $2$ & $3$ & $5$ & $1$ & $4$ \\
      $3$ & $5$ & $4$ & $2$ & $1$ \\
      $5$ & $4$ & $1$ & $3$ & $2$ \\
      $4$ & $1$ & $2$ & $5$ & $3$
    \end{tabular}
  \end{itemize}

  Notice that ${\rm Aut}(R)$ is the cyclic group generated
  by the permutation $(1 \; 2 \; 3 \; 4 \; 5)$, while ${\rm Aut}(S)$ is the
  cyclic group generated by the permutation $(1 \; 2 \; 3 \; 5 \; 4)$:
  \begin{equation*}
    {\rm Aut}(R) = \begin{tabular}{l}
      {\rm Identity} \\
      {\rm ($1$ $2$ $3$ $4$ $5$)} \\
      {\rm ($1$ $3$ $5$ $2$ $4$)} \\
      {\rm ($1$ $4$ $2$ $5$ $3$)} \\
      {\rm ($1$ $5$ $4$ $3$ $2$)}
    \end{tabular}
    \qquad
    {\rm Aut}(S) = \begin{tabular}{l}
      {\rm Identity} \\
      {\rm ($1$ $2$ $3$ $5$ $4$)} \\
      {\rm ($1$ $3$ $4$ $2$ $5$)} \\
      {\rm ($1$ $5$ $2$ $4$ $3$)} \\
      {\rm ($1$ $4$ $5$ $3$ $2$)}
    \end{tabular}
  \end{equation*}

  From ${\rm Aut}(R)$ and ${\rm Aut}(S)$ we can easily obtain ${\rm CP}(R) =
  {\rm OP}(R) = {\rm CP}(S) = {\rm OP}(S) = \left\{ \left\{\{1\}, \{2\}, \{3\},
      \{4\}, \{5\}\right\} , \left\{ \left\{1, 2, 3, 4, 5\}\right\} \right\} \right\}$.
  Clearly, $S$ is not expressible from $R$, since we have $D(R) = D(S)$ but
  ${\rm Aut}(R) \not\subseteq {\rm Aut}(S)$; on the other hand,
  ${\rm OP}(R) \subseteq {\rm OP}(S)$ and $D(S) \subseteq D(R)$, and ${\rm
    CP}(R) \subseteq {\rm CP}(S)$  and         $D(S) \subseteq D(R)$.
  The fact that $S$ is not expressible from $R$ can be correctly determined
  through orbit partitions or through cycle partitions by observing that:
  ${\rm OP}(\{R, S\}) = \Big\{\big\{\{1\},\{2\},\{3\},\{4\},\{5\}\big\}\Big\}
  \neq {\rm OP}(R)$
  or
  ${\rm CP}(\{R, S\}) = \{\{\{1\},\{2\},\{3\},\{4\},\{5\}\}\}
  \neq {\rm CP}(R)$.
\end{Example}

\section{Expressiveness in graph-based data bases}

In this section we study a simple graph-based model  where two labeled graphs are
used to model data bases.
A data base consists of two distinct layers: a \emph{schema} layer and a
\emph{structure} layer; the objects can be found in the latter, while the former
describe the data organization. Each layer is a labeled weakly-connected
directed graph, moreover there exists a function that maps a schema into a
structure: such function will be called an \emph{extension}. Both vertices and
edges of the graphs are labeled, and we can assume that the sets of edge labels
and vertex labels, as well as schema labels and structure labels, are disjoint.
An example of data base is represented in
Figures~\ref{fig:schema},~\ref{fig:structure}, from which  it is easy to note
how the schema and the structure are
closely related, the following definitions only formalize the intuitive idea.

\begin{figure}[htb]
  \begin{center}
    \includegraphics[width=\textwidth]{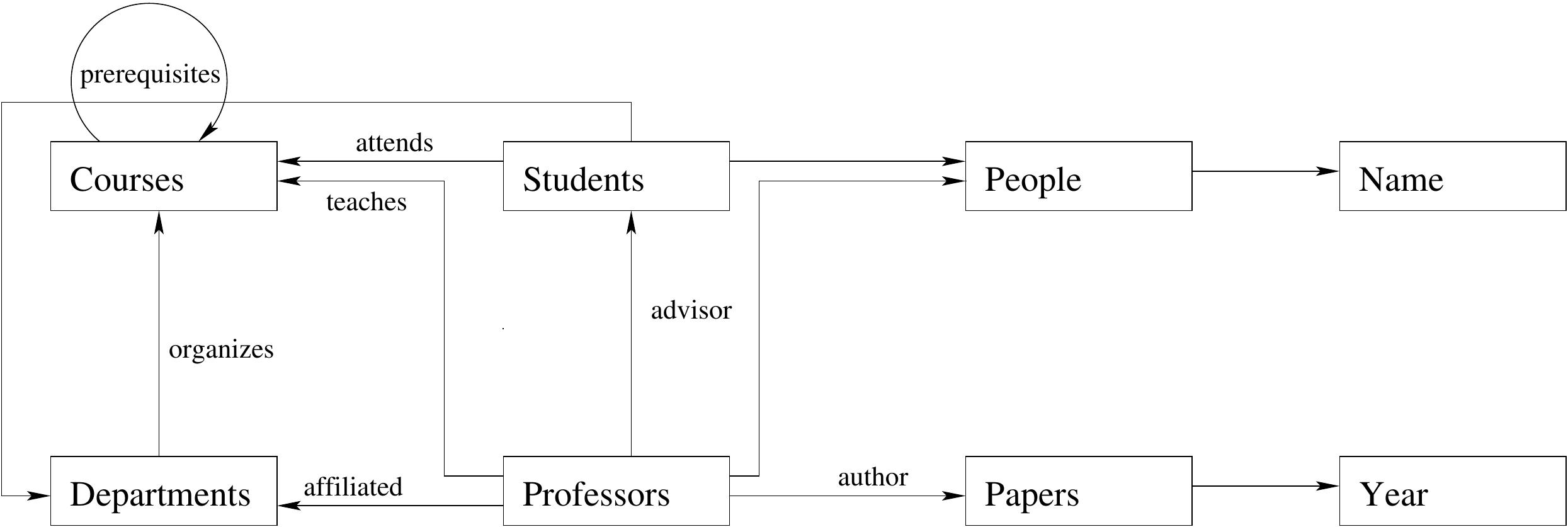}
  \end{center}
  \caption{Example of schema}
  \label{fig:schema}
\end{figure}
\begin{figure}[htb]
  \begin{center}
    \includegraphics[width=\textwidth]{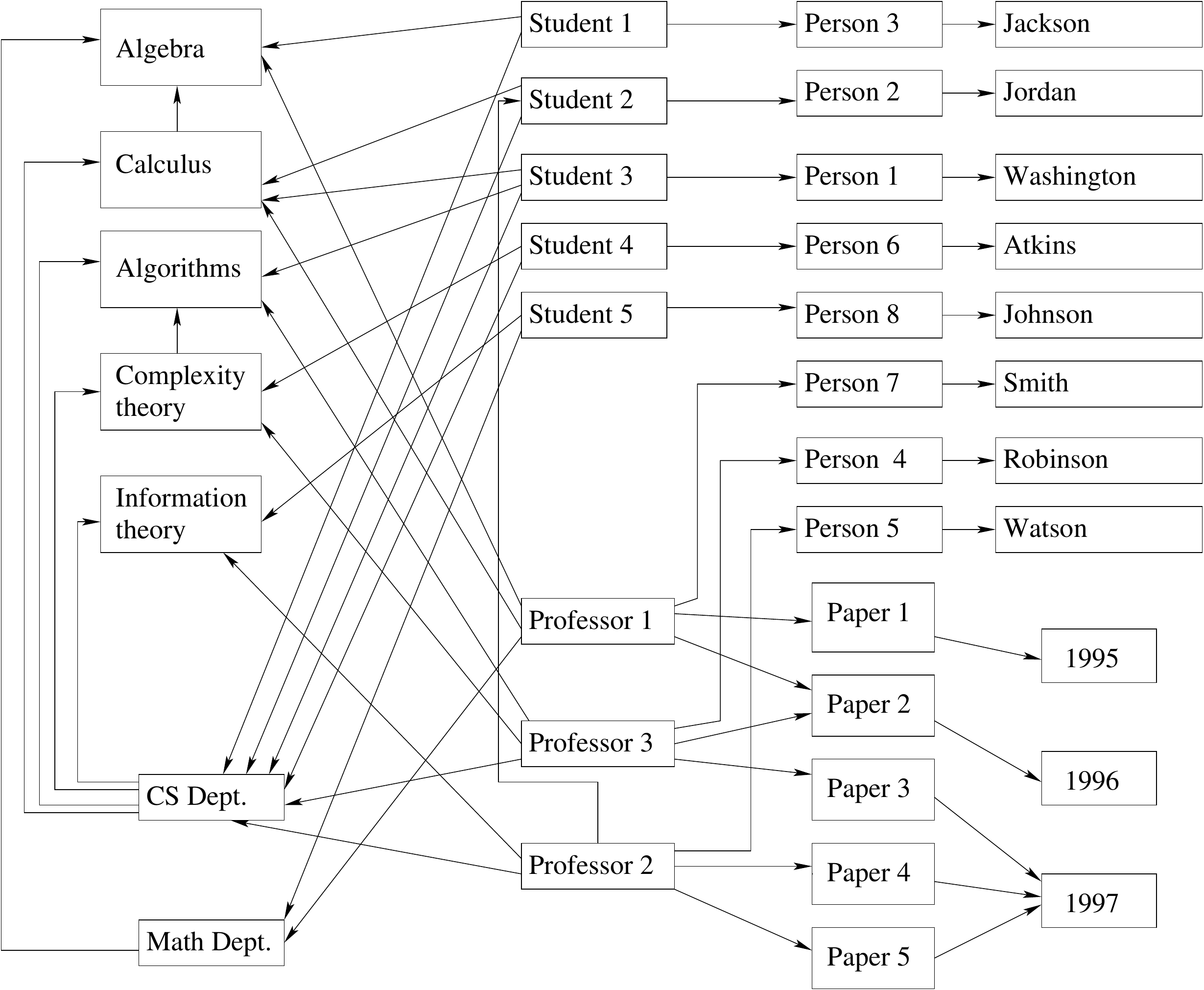}
  \end{center}
  \caption{Example of structure}
  \label{fig:structure}
\end{figure}

\begin{Def}[Schema]
  A  {\em schema graph}, in short {\em schema},  is
  a triple   $\Sigma=(G, \lambda_1, \lambda_2)$,
  where $G= (V,E)$ is  an oriented, weakly-connected
  graph, and  $\lambda_1$, $\lambda_2$ are respectively the injective functions that maps each node
  (resp. edge) to its label.
\end{Def}

\begin{Def}[Structure]
  \label{def_struttura}
  A  {\em structure}
  is a triple ${\struct} =(S, \lambda_1', \lambda_2')$, with
  $S$ a colored oriented graph $S=(V,E,\mu)$,
  where $V$ is the set of {\em nodes} of the structure,
  $E\subseteq V\times V$  is the set of {\em edges},
  $\lambda_1'$, $\lambda_2'$ are respectively the injective functions that maps each node
  (resp. edge) to its label, and $\mu : E \to \Gamma $,
  is a  labeling of the edges over the finite alphabet $\Gamma$, called {\em coloring}
  of the structure.
\end{Def}

In the following, we will use the set $\Gamma =\{true, false\}$ of
colors that allows to specify that a link between object instances
in ${\struct}$ is actual or not. In the example of
Fig.~\ref{fig:structure}, only the links labeled true are represented, and the
presence (or the abscence) of links labeled false does not change the data
stored in the data base. In Fig.~\ref{fig:false} is represented a
part of the structure, where false links are represented with dotted arrows.

The schema and the structure must be strongly correlated; in fact there must
exist a function, called \emph{extension} (denoted by $\Inst$), mapping the schema into the
structure. In order to have a sound definition of extension some restrictions
must be enforced, as pointed out in the following definition, where $\Pow(A)$
stands for the family of all nonempty subsets of $A$.

Informally $\Inst$ maps each vertex of the schema into some vertices of the
structure and each edge  of the schema into some edges of the
structure.

\begin{figure}[htb]
  \begin{center}
    \includegraphics[width=10cm]{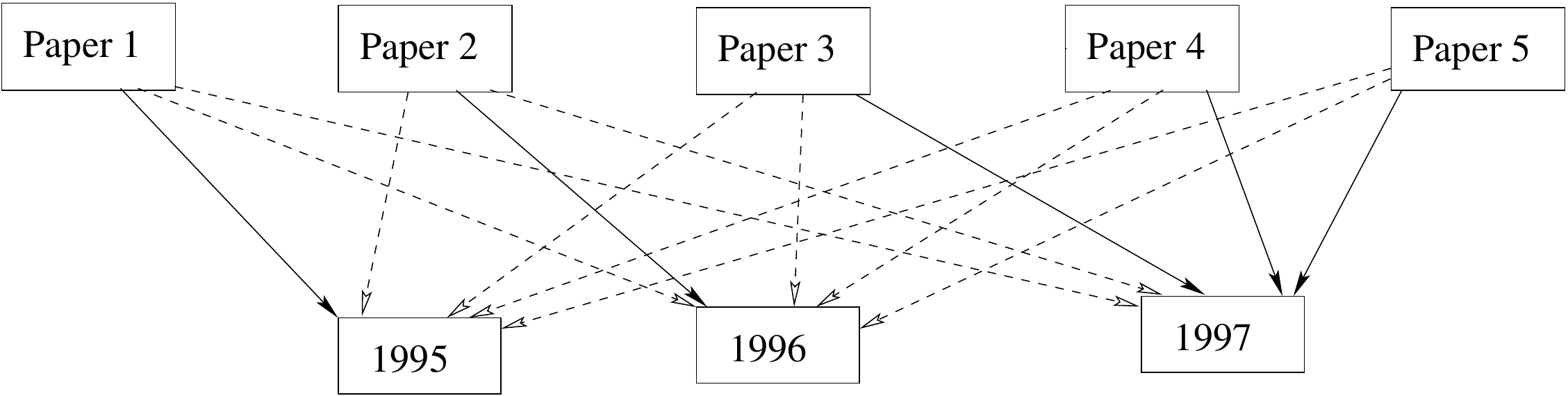}
  \end{center}
  \caption{Example of structure with false links}
  \label{fig:false}
\end{figure}

\begin{Def}[Extension]
  \label{def_f1}
  Let $\Sigma=(G=\langle V,E\rangle, \lambda_1,\lambda_2)$ be a schema  and
  ${\struct}=(S, \lambda_1',\lambda_2')$ a structure, where
  $S=(V',E',\mu)$.
  Then ${\struct}$ is an
  {\em extensional structure} of $\Sigma$
  if there  is a   function (the  {\em extension}) from $\Sigma$ to ${\struct}$,
  $\Inst : V \to Pow(V')$,
  such that:
  \label{condizioni_Rel}
  \begin{enumerate}
  \item
    $\{\Inst(v) : v\in V(G)\}$ is  a partition $\{V_1, \cdots, V_n\}$  of the  set $V'$,

  \item
    for every $x \in V_i, y \in  V_j$, the pair
    $(x,y) \in E'$ iff
    $(\Inst^{-1}(x),  \Inst^{-1}(y)) \in E$;
  \end{enumerate}
\end{Def}

Notice that the first point of the definition of extension implies that the
function $\Inst^{-1}$ is well  defined.
In the following, if  ${\struct}$ is  the extensional structure of  $\Sigma$, then
we write  ${\struct}= \Inst(\Sigma)$ and we will simply say that
${\struct}$ is a structure of $\Sigma$.
Given two vertices $v_1$ and $v_2$ of the schema, connected with a
link $(v_1,v_2)$ then in the structure there must exist all links $(w_1,w_2)$
for $w_1\in \Inst(v_1)$,  $w_2\in \Inst(v_2)$. Such requirement justifies the
introduction of a labeling (and especially of a true-false labeling) in order to
have a reasonable graph-based model.

\begin{Def}[Data base]
  \label{basedati}
  A  {\em data base} $B$  is a pair $(\Sigma,{\struct})$,
  where $\Sigma$ is a schema and
  ${\struct}$ is an extensional structure of  $\Sigma$.
\end{Def}

The schema describes the conceptual organization of the data, while the data
content or instantiation of the data base is given by the extensional structure
associated to the schema.


It is not hard to notice that, given a schema, there is a one-to-one
correspondence between structures and extension functions, therefore we
will sometimes use the pair $(\Sigma,\Inst)$ as a data base.

Some preliminary definitions are required for introducing our query language.
Given a partial function $f:A\mapsto B$ (i.e. a function where each element of
$A$ can be associated to one or none of the elements of $B$), by $\dom(f)$ we
denote the domain of
$f$, that is the set of elements $x\in A$ such that $f(A)$ is defined.  Let
$f,g$ be two partial functions from the set $A$ to the set $Pow(B)$.  Then {\em
  $f$  is a restriction of $g$}, denoted by $f \leq g$, if $\dom(f)\subseteq \dom(g)$ and for
every $x\in \dom(f)$, $f(x)\subseteq g(x)$.
Moreover by $\Imm(\f)$ we denote the
set obtained as union of all images of elements in $\dom(\f)$:
formally $\Imm(\f) = \cup_{x \in \dom(\f)}\f(x)$.

\begin{Def}[Instance]
  \label{morfi}
  Let $B=(\Sigma,\Inst)$ be a data base.
  An \emph{instance} of $B$ is a restriction $f$ of $\Inst$ such that $\dom(f)$
  induces a weakly-connected subgraph of $\Sigma$.
\end{Def}

The following notations will be used in the rest of the paper.
The set   ${\I}(B)$ is the set of all
instances of $B$.
Let ${\mathcal \I}$ be a subset of ${\I}(B)$, then by  $\Imma({\mathcal \I})$,
we mean the set of nodes of the structure of $B$ that is the union of
all images of instances in ${\I}$, while
$\Dom({\mathcal \I})$
is the union of all domains of instances in ${\I}$.
An element in $\Imma({\I})$ is called a {\em value}, while an element in
$\Dom({\I})$ is called a {\em name}.
Then the image of a name $x \in \Dom({\I})$ is the subset
$A$ of $\Imma({\I})$ such that $A=\cup_{{\f}_i \in {\I}}{\f}_i(x)$.
For a value $y\in\Imma({\mathcal \I})$, the {\em  inverse image of $y$},  denoted
by $\name(y)$, is the  name of $\Dom({\I})$ that
is mapped by $\Inst$,  to a set containing the
element $y$.
Similarly, given a set $A$ of values, the {\em inverse image of $A$}
is the set $\names(A)$ which is union of all inverse images of
the values  in $A$.

\subsection{The graph  algebra}
\label{sec:graph-algebra}

Our graph data model  is proposed as a domain-preserving data base, along the
same lines as other papers where the expressiveness of the relational algebra is
studied~\cite{Par78,Banc1}, and
it gives a formal embedding for languages used for the retrieval
of graph-structured information~\cite{LZ}.
The requirement that we are dealing with domain-preserving data bases reflects
in the query language: in fact we have no
operation for creating new elements or modifying the schema graph, and all
operations must preserve the schema and the original structure.

The main consequence of the assumption that our model is domain preserving
consists in the fact that we will deal with a schema  which is mapped to an
instance through an extensional mapping. Therefore there is a complete
equivalence between subgraphs of the structure and restrictions of the
extensional mapping. We are now able to introduce the operations of our graph
algebra: according to our reasoning above we can describe the operation as over
partial functions whenever it allows a simpler formulation.

\begin{Def}[Addition]
  \label{sovrapposizione}
  Let $B=(\Sigma,{\struct})$ be a data base and let $f_1,f_2 \in {\I}(B)$.
  The {\em Addition} of $f_1$ and $f_2$, denoted as $f_1 \oplus f_2$, is the following
  function over domain $\dom(f_1)$. The operation is defined only if $\dom(f_1)=
  \dom(f_2)$:
  $$(f_1\oplus f_2)(x)=f_1(x)\cup f_2(x)$$
\end{Def}

\begin{Def}[Product]
  \label{prodotto}
  Let $B=(\Sigma,{\struct})$ be a data base, and let $f_1,f_2$ be two functions in
  ${\I}(B)$.
  The {\em Product}
  of $f_1$ and $f_2$, denoted as $f_1\otimes f_2$ is the
  instance in ${\I}(B)$ defined as follows:

  \vspace{-1em}
  \begin{equation}
    (f_1 \otimes f_2)(x)
    =\left\{ \begin{array}{r@{\quad }l}
        f_1(x) \cap f_2(x) & \text{if } x \in \dom(f_1) \cap  \dom(f_2), f_1(x) \cap f_2(x) \not= \emptyset\\
        f_2(x) &  \text{if }x \in \dom(f_2) - \dom (f_1)\\
        f_1(x) &  \text{if }x\in \dom(f_1) - \dom(f_2)\\
        \text{undefined} & \text{otherwise}
      \end{array}
    \right.
  \end{equation}
  The product is defined only if $Dom(f_1\otimes f_2)$ induces a weakly-connected
  subgraph of $\Sigma$.
\end{Def}

\begin{Def}[Projection]
  \label{proiezione}
  Let $B=(\Sigma,{\struct})$ be a data base.
  Let $f$ be a function in ${\I}(B)$ and
  let  $A$ be a subset of the domain of $f$,
  such that  $A$ induces in $\Sigma$ a weakly-connected subgraph. The
  {\em projection of $f$ on $A$}, denoted as
  $\Pi_A (f)$, is the instance defined as follows:

  \vspace{-0.5em}
  \begin{equation}
    \Pi_A (f)(x) = \left\{ \begin{array}{r@{ \quad}l}
        f(x) & \text{if }x \in A\\
        \text{undefined} &\text{if } x \notin A
      \end{array}
    \right.
  \end{equation}
\end{Def}

\begin{Def}[Difference]
  \label{differenza}
  Let $B=(\Sigma,{\struct})$ be a data base.
  Let $f_1,f_2$ be two functions in ${\I}(B)$ over the same
  domain $A$.
  The {\em difference of $f_1$
    by  $f_2$}, denoted as $f_1 \ominus f_2$ is the following instance:

  \begin{equation}
    f_1 \ominus f_2 (x) = \left\{ \begin{array}{r@{ \quad}l}
        f_1 (x) - f_2 (x) & \text{if } x \in A, f_1(x) - f_2(x) \not=\emptyset\\
        \text{undefined} & \text{otherwise}
      \end{array}
    \right.
  \end{equation}
  The difference is defined only if $Dom(f_1 \ominus  f_2)$ induces a weakly-connected
  subgraph of $\Sigma$.
\end{Def}

Since the coloring of the edges encodes the fact that a relation between two
objects is actual or not, it is natural that the query language has some tools
for exploiting such coloring. In our model we will need to extract instances
where ``similar'' edges are the same color. The definition of \emph{selector} is
the first step in such direction.

\begin{Def}[Selector]
  Let $\Sigma$ be  the schema of a data base $B$.
  Then a {\em selector of $\Sigma$} is a pair $(G, \sigma)$ consisting of
  a weakly-connected subgraph $G$ of $\Sigma$ and a coloring $\sigma: E \to \Gamma$
  of the edges of $G$.
\end{Def}

Querying for a selector in a data base returns all subgraphs of the structure
that are isomorphic to the selector: each such subgraph is indeed called a
\emph{simple instance}.
Moreover, it is natural to  define an operation  of selection that
allows to obtain instances which are compatible with a coloring
of the schema over the alphabet $\Gamma$. This is the last operation of our algebra.

\begin{Def}[Simple instance]
  \label{simple}
  Let $(\Sigma,\Inst)$ be a data base, where
  $\struct=(V',E', \mu)$.
  Let   $(G_s,\sigma)$ be a selector of  $\Sigma$,
  where $G_s=(V(G_s),E(G_s))$.
  Then a {\em simple instance induced by the selector $(G_s, \sigma)$} is
  a restriction $f$ of $\Inst$ such that $\Dom(f)=V$, $|f(v)|=1$ for each
  $v\in\Dom(f)$ and $\mu(\f(x) ,\f(y)) = \sigma(x,y)$ for each  $(x,y) \in E(G_s)$.
\end{Def}

Notice that all simple instaces have the same domain.

\begin{Def}[Selection]
  \label{selezione}
  Let $B=(\Sigma,{\struct})$ be a data base.
  Let $f$ be a function in ${\I}(B)$
  and $(G_s, \sigma)$  a selector.
  Let ${\mathcal F}$ be the set of all simple instances
  induced by $G_s$ that are also subinstances of $f$.
  The {\em selection of $\f$ by $(G_s, \sigma)$}, denoted as
  $\f | (G_s,\sigma)$, is $\bigoplus_{g \in {\mathcal F}} g$.
\end{Def}

\section{Stability}

Given a set ${\mathcal \I}$ of instances, our first aim is to give a
characterization of all instances
that can be obtained with a query that uses only the information
contained in the instances in ${\mathcal \I}$, or equivalently by an expression of
the algebra that has only instances in ${\mathcal \I}$ as operands.
In such direction the main result of this section is that expressiveness in our
graph algebra is equivalent to the conservation of a certain partition. It is natural to
associate a notion of \emph{undistinguishability} to a partition, where all
elements in a set of the partition are deemed undistinguishable.
We share the goals of~\cite{Par78}, but we have introduced in this paper a new
framework, that  is we are looking for a notion of
expressiveness that is coherent with our meta theorem.
Just as the notion  of {\em automorphisms}, introduced in~\cite{Par78,Banc1} for
relations,
gives a global description of the logical dependencies among data that must be
preserved when querying the data base,  in our model a partition (or an
equivalence relation) will represent all such logical dependencies.
The equivalence partition over elements of the structure that we will
study is called \emph{stability} and is denoted by $\stronguniform{}{}{\I}$
(where $\I$ is an instance).

The simplest possible form of undifferentiation (called \emph{0-stability}) is
based on the idea that we are able to distinguish images of different vertices
of the instance and
vertices of the extensional structure belonging to different functions of
$\I$. Such notion basically consists of using expressions in our algebra that do
not contain any selection.

\begin{Def}
  \label{def:split}
  Let $A$ be a subset of the image of a set ${\I}$ of
  instances.
  Then $A$ is {\em split} by ${\I}$ iff there is a
  function  $f_i$ in ${\I}$ such that
  $A \cap \Imm({\f}_i) \neq \emptyset$ and
  $A - \Imm({\f}_i) \neq \emptyset$.
\end{Def}

We are now able to introduce formally the definition of 0-stability, as follows:

\begin{Def}[0-stable]
  \label{0-stable}
  Let ${\I}$ be a set of instances over a data base $(\Sigma,{\struct})$, and
  let $A$ be a subset  of $\Imma({\mathcal \I})$.
  Then $A$ is \emph{0-stable} w.r.t.  ${\mathcal \I}$, if the two following
  conditions hold:
  \begin{enumerate}
  \item
    $A \subseteq \Imm(\f(x))$, where $\f\in\I$, and $x\in \Dom({\I})$ is a name of
    the schema.
  \item
    for each function  $f\in\I$, then $A$ and $\Imm({\f}_i)$ are disjoint or one is
    contained into the other one.
  \end{enumerate}
\end{Def}

A more refined notion of undifferentiation is called {\em $1$-stability};
informally a set $A$ is 1-stable w.r.t. $B$ if $A$ is 0-stable and $B$ is not
able to distinguish two vertices of $A$ with edges outgoing from $B$ and ingoing
in $A$ (or outgoing from $A$ and ingoing in $B$).
Notice that 1-stability is a binary relation  over
subsets of $\Imma({\I})$, while 0-stability is a unary relation. The formal
definition is:

\begin{Def}[1-stable]
  \label{un_passo}
  Let ${\I}$ be a set of instances over a data base
  $(\Sigma,{\struct})$ and let $A$ and $B$ be two disjoint
  subsets $\Imma({\mathcal \I})$.
  Then $A$ is  {\em 1-stable} w.r.t.  $B$ and  ${\mathcal \I}$,
  denoted as  $\stronguniform{A}{B}{1,\I}$ if the following conditions are verified:
  \begin{enumerate}
  \item
    $A$ is 0-stable w.r.t. ${\I}$;
  \item
    for each edge  $(a_1, b_1)$  of $\struct$, with $a_1\in A, b_1\in B$ and
    for each $a_2\in A$ there exists $b_2\in B$ such that $\mu(a_1, b_1)=\mu(a_2, b_2)$;
  \item
    for each edge  $(b_1, a_1)$  of $\struct$, with $a_1\in A, b_1\in B$ and
    for each $a_2\in A$ there exists $b_2\in B$ such that $\mu(b_1, a_1)=\mu(b_2, a_2)$.
  \end{enumerate}
\end{Def}

Informally 1-stability means that whenever there is an colored edge (say a red
edge) from a vertex of  $A$ to a vertex of $B$, then all vertices of $A$ have a
red edge ingoing in $B$. In other words if we assume that $B$ is
undistinguishable, then also $A$ is undistinguishable, by any single-edge path.
The notion of  1-stability can be further generalized,
but first we need a new definition.

\begin{Def}[Path]
  \label{cammino}
  Let $G=(V,E)$ be a labeled graph.
  Then a {\em colored path} in  $G$ is a pair $(p,s)$
  where $p=<v_1, e_1, v_2, \ldots, v_{l-1}, e_{l-1}, v_l>$,
  and for every $1 \leq i \leq l$, $v_i$ belongs to $V$ and
  $e_i$ is an edge of $G$ such that
  $e_i = (w_i, w_{i+1})$  or $e_i = (w_{i+1},w_i)$.
  Moreover $s$ is the sequence $<\mu(e_1), \ldots,\mu(e_{l -1})>$ where $\mu(e_i)$ is the
  color of the edge $e_i$  in $G$.
\end{Def}

Notice that the definition of path used in the paper is different from the one
that can be usually found in a graph theory textbook, as arcs can also be in the
reverse direction.
The \emph{length} of a path is the number of edges it contains.
Let $B= (\Sigma,{\struct})$ be a data base and let $(p,s)$ be a colored path of
$\struct$, with  $p=<v_1,e_1,v_2, \ldots, e_n,v_{n+1}>$.
Then the {\em path schema} of $(p,s)$
is the pair $(p',s)$, with $p'=<v'_1,e'_1, v'_2, \ldots,e'_n, v'_{n+1}>$, where
for every $1 \leq i \leq {n+1}$, $v'_i=\Inst^{-1}(v_i)$ and
$e'_i\Inst^{-1}(e_i)$.

\begin{Def}
  \label{path-dependencies}
  Let $x, y$ be two nodes of the ${\mathcal \I}$-structure, and let $Z$ be a
  subset of $\Imma({\mathcal \I})$, then the \emph{path dependencies} from $x$ to
  $y$ in $Z$,
  denoted as ${\AM}_{k,Z, {\mathcal \I}}(x,y)$ is  the  set of   path schemata of all paths
  of the ${\mathcal \I}$-structure that are starting in $x$ and ending in $y$ and
  entirely contained in $Z$.
\end{Def}

Informally given $x,y,Z$, their path dependencies is obtained by removing all
vertices not in $Z$, then computing all
possible paths from $x$ to $y$, and finally computing the respective path
schemata.
The next step is to generalize 1-stability to $k$-stability, that is taking into
account length-$k$ paths instead of simple edges (that is length-1 paths).

\begin{Def}[$k$-stable]
  \label{k_passi}
  Let ${\I}$ be a set of instances over a data base
  $(\Sigma,{\struct})$, let $k$ be an integer larger than one, and let $A$ and $B$
  be two disjoint
  subsets of $\Imma({\mathcal \I})$.
  Then $A$ is  {\em $k$-stable} w.r.t.  $B$ and  ${\mathcal \I}$,
  denoted as  $\stronguniform{A}{B}{k\I}$ if the following conditions are verified:
  \begin{enumerate}
  \item
    $A$ is 0-stable w.r.t. ${\I}$;
  \item
    $A$ is $(k-1)$-stable w.r.t.  $B$ and  ${\mathcal \I}$;
  \item
    for each $a_1, a_2 \in A,  b_1\in B$ there exists $b_2\in B$ such that
    $\AM_{k,A\cup B,{\mathcal \I}}  (a_1,b_1) \subseteq \AM_{k,A\cup B,{\mathcal  \I}}  (a_2,b_2)$.
  \end{enumerate}
\end{Def}

The main idea is that when $\stronguniform{A}{B}{k\I}$ then if $B$ is
undistinguishable also $A$ is undistinguishable when only paths no longer than
$k$ are taken
into account. Our main definition follows:

\begin{Def}[Stability]
  \label{rel_collettiva}
  Let  ${\mathcal \I}$ be  a set of instances, and let  $A,B$ be two disjoint
  subsets of $\Imma({\mathcal \I})$. Then $A$ is {\em stable w.r.t. $B$} in
  ${\mathcal \I}$, denoted as $\stronguniform{A}{B}{\I}$,  if
  $\stronguniform{A}{B}{k, \I}$
  for all $k\in \mathbb{N}$.
\end{Def}

By  Def.~\ref{rel_collettiva}, it is immediate
to verify the following properties of stability:

\begin{Lem}
  \label{prop1_rel}
  Let  ${\mathcal \I}$ be a set of instances, and let  $A,B,C \subseteq
  \Imm({\mathcal \I})$, then:
  \begin{enumerate}
  \item
    if $\stronguniform{A}{B \cup C}{\I}$ and  $\names(B) \cap \names(C)=\emptyset$, then
    $\stronguniform{A}{B}{\I}$ and $\stronguniform{A}{C}{\I}$,
  \item
    if  $\names(B) =\names(C)=\{x\}$, $\stronguniform{A}{B}{\I}$ and
    $\stronguniform{A}{C}{\I}$, then $\stronguniform{A}{B\cup C}{\I}$,
  \item
    if  $\names(B) =\names(C)=\{x\}$, $B \cap C \not= \emptyset$, $\stronguniform{B}{A}{\I}$ and
    $\stronguniform{C}{A}{\I}$, then $\stronguniform{B\cup C}{A}{\I}$.
  \end{enumerate}
\end{Lem}

Stability is a relation between disjoint subsets of the domain. The definition
of expressiveness in the query language that we want to obtain is based on
partitions, and now we are able to introduce the class of partitions we are
interested into.
A  partition ${\mathcal P}$ of nodes of the structure is called \emph{valid} if
and only if for each set $A$ of the partition and every set $B$ that is a
union of sets of  ${\mathcal P}$, then $B$ cannot differentiate $A$.

\begin{Def}
  \label{partizione_accettabile}
  Let   ${\mathcal \I}$ be a set of instances.
  A partition ${\mathcal P}=\{P_1, \ldots,P_k\}$ of
  $\Imm({\mathcal \I})$ is  {\em valid} if  for every $P_i\in {\mathcal P},\,
  L\subset \{1,\ldots ,k\},\, L\neq \emptyset ,\, i\notin L$, then
  $\stronguniform{P_i}{\bigcup_{l\in L} P_l}{\I}$.
\end{Def}

Given a set ${\I}$ of instances, then there may be various
valid partitions of $\Imma({\mathcal \I})$, and at least one valid partition
always exists (the partition where each vertex of the extentional structure is a set).
Some valid partitions are more representative of the actual undifferentiation,
in fact we will  assume  as a measure of the undifferentiation
induced by ${\mathcal \I}$  the coarsest valid partition, which
we will call {\em canonical partition} and denote as ${\mathcal C}_{\mathcal \I}$.
We can show that the definition of canonical partition is well-formed.

\begin{Thm}
  Every set ${\mathcal \I}$ of instances has a unique
  canonical partition  ${\mathcal C}_{\mathcal \I}$.
\end{Thm}

\begin{proof}
  Clearly the partition of $\Imma({\mathcal \I})$ into singletons is a valid
  partition, so there exists at least one canonical partition. Now assume to the
  contrary that there exist two coarsest valid partitions ${{\mathcal P}_1}$ and
  ${{\mathcal P}_2}$.
  Let $R_1$ and $R_2$ be the equivalence relations induced by the partitions
  ${{\mathcal P}_1}$ and ${{\mathcal P}_2}$, respectively.
  Let $R^*$ be the  transitive closure of the relation $R$
  defined as follows:
  $x R y$ if and only if $x$ and $y$ are in the same set of ${{\mathcal P}_1}$ or
  ${{\mathcal P}_2}$.  We can prove that the partition
  ${\mathcal P}$ induced by $R^*$ is a valid one of index strictly less than $k$.
  By construction of $R^*$ each set of ${\mathcal P}$ is a union of sets in
  ${{\mathcal P}_1}$ and also a union of sets in ${{\mathcal P}_2}$, moreover each
  set in ${\mathcal P}$ is contained in the image of a single name (since each set
  must be 0-stable). Notice that $R^*\neq {\mathcal P}_1$ iff ${\mathcal
    P}_1\neq {\mathcal P}_2$.

  Let $X_i$ be a set of  ${\mathcal P}$, and let $Z_k$ be a class of
  ${{\mathcal P}_1}$ contained in $X_i$. Since  ${{\mathcal P}_1}$ is a valid
  partition, and $X_i$ is a union of disjoint sets of ${\mathcal P_1}$, by
  Lemmata~\ref{rel_collettiva},~\ref{prop1_rel}, we have that
  $\stronguniform{X_i}{Z_i}{\I}$ and
  $\stronguniform{Z_i}{X_i}{\I}$. Let $X_i$, $X_j$ be two sets of  ${\mathcal P}$, with
  $X_i=Z_{i_1}\cup\cdots Z_{i_k}$ and $X_j=Z_{j_1}\cup\cdots Z_{i_l}$. We have
  already proved that $\stronguniform{X_i}{Z_j}{\I}$ and  $\stronguniform{Z_j}{X_i}{\I}$,
  applying again Lemmata~\ref{rel_collettiva},~\ref{prop1_rel} and noting that
  $X_j=Z_{j_1}\cup\cdots Z_{i_l}$ we obtain $\stronguniform{X_i}{X_j}{\I}$, By the
  generality of $X_i$ and $X_j$ the partition  ${\mathcal P}$ is valid. Moreover
  ${\mathcal P}$ is coarser than  ${\mathcal P_1}$, which is a contradiction.
\end{proof}


\section{Expressiveness}

In this section we will prove our main result regarding the expressiveness of
the graph-based query language by showing that  a function can be computed if
and only if adding such function does not change the canonical partition. In the
following we will assume that the union of the images of all functions in $\I$
is exactly the universe set; such assumption does not violate the generality
since otherwise we would simply have some sets of the canonical partition whose
union consists of exactly those elements of the universe set that are not in any
function in $\I$.

\begin{Thm}
  \label{graal-expressivess}
  Let $\BI({\mathcal \I})$ be the set of functions that are a result of an
  expression of the graph algebra where only
  functions of ${\mathcal \I}$ are operand. Then $f\in \BI({\mathcal \I})$ if and
  only if the canonical partition induced by ${\mathcal \I}$ is equal to the
  canonical partition induced by ${\mathcal \I}\cup \{f\}$, that is ${\mathcal
    C}_\I={\mathcal C}_{\I\cup \{f\}}$.
\end{Thm}


The following two properties, that are consequences of
Def.~\ref{partizione_accettabile}, will be useful to
prove the main result of the paper.


\begin{Prop}
  \label{f_1}
  Let ${\mathcal \I}$ be a set of instances and let ${\mathcal P}$
  be a valid  partition.
  Then the image of  every instance $\g \in \BI({\mathcal \I})$ is
  the union of sets of ${\mathcal P}$.
\end{Prop}

\begin{proof}
  We prove the lemma by induction on the number $n$ of operations of the
  expression for $\g$.
  If  $n=0$, then $\g$ is a function in $\mathcal \I$.
  Since all sets of a valid partition are $0$-stable, no set of a valid partition
  has both an element in $\Imm(\g)$ and an element not in $\Imm(\g)$, therefore
  the union of all sets of ${\mathcal P}$ that are contained in
  $\Imm(\g)$ is contained in $\Imm(\g)$. To prove that such containment is not
  strict (i.e. such union is equal to  $\Imm(\g)$) it is sufficient to note that
  all elements of $\Imm(\g)$ belong to some set of ${\mathcal P}$.

  Assume now that $n > 0$ and $\g$ is obtained by the application
  of an operation to two expressions  $\f_1$ and $\f_2$ in $\BI({\mathcal \I})$,
  or one expression for a selection.
  Clearly, by inductive hypothesis the images of $\f_1$ and $\f_2$
  are obtained as the union of some sets of ${\mathcal P}$.
  It is immediate to verify the lemma for the case that
  $\g=\f_1 \oplus \f_2$,
  $\g=\f_1 \otimes  \f_2$,
  $\g=\f_1 \ominus \f_2$ and
  $\g=\Pi_A(\f_1)$.
  Finally, assume that $\g= \f_1|G_s$, where $G_s$ is a selector.
  By definition of the selection, then $\Imm(\g)$ is the union of
  the images of all simple instances induced by $G_s$.
  Assume to the contrary that there exists a set $A$ of the valid
  partition such that $A \not\subseteq \Imm(\g)$ and $A \cap \Imm(\g) \not= \emptyset$.
  Now, let $y \in A - \Imm(\g)$ and $x \in A \cap \Imm(\g)$. Hence, by
  construction of  selection,
  $y$ cannot be in the image of any simple instance induced by $G_s$, while $x$ is
  contained in the image of a simple instance induced by
  $G_s$.
  By inductive hypothesis  $\Imm(\f_1)$ is union of sets of  ${\mathcal P}$,
  moreover since $A$ is a set of  ${\mathcal P}$, also $\Imm(\f_1) - A$ is union
  of sets of the valid partition, implying that $\stronguniform{A}{\Imm(\f_1) - A}{\I}$,
  It follows that, for each $z \in \Imm(\f_1) - A$,
  $\AM_{\Imm(\f_1), {\mathcal \I}}(x,z) \subseteq \AM_{\Imm(\f_1), {\mathcal \I}}(y,v)$
  for some $v \in \Imm(\f_1) - A$.
  This implies that there is a simple instance induced by $G_s$ that
  has in its image $y$, which is a contradiction with
  the above assumption.
  Consequently, the image of $\g$ must be union of sets of ${\mathcal P}$.
\end{proof}

We will prove that  an alternative characterization of
canonical partition is as the partition induced by the equivalence relation
$R_\I$ between elements of $\Imm(\I)$, where $x R_{\I} y$ if and only if
for every instance $\f \in \BI({\mathcal \I})$, $x\in\Imm(\f) \Leftrightarrow y\in \Imm(\f)$.
In the following of the paper let $\PBI{\I}$ denote the partition induced by the
equivalence relation $R_\I$.
Successively we will prove that a function $\f$ belongs to $\BI({\mathcal \I})$
if and only if $\Imm(\f)$ can be obtained as union of sets of  $\PBI{\I}$,
completing the proof of our main result, in two steps. First we will
prove that a function $f$ belongs to $ \BI({\I})$ if and only if  $\Imm(f)$ is
union of sets in $\PBI{\I}$, then we will prove that  $\PBI{\I}=\mathcal{C_\I}$.
The following proposition is an immediate consequence of the definitions of
$x R_{\I} y$ and of projection.

\begin{Prop}
  \label{unique1}
  Let $x,y\in\Imm(\I)$ such that  $x R_{\I} y$. Then both $x$ and $y$ belong to the
  set $\Inst(z)$ for some name $z$.
\end{Prop}

\begin{Cor}
  \label{unique}
  Let $P^{\BI}_{\mathcal \I}$ be the  partition induced by a set ${\mathcal \I}$
  of instances.
  Then, the inverse image of every set of the partition consists
  of a single  vertex.
\end{Cor}

\begin{Lem}
  \label{minimo_ottenibile}
  Let ${\mathcal \I}$ be a set of instances over $B$,  let $\PBI{\I}$
  be the partition induced by ${\I}$, and let $A\in P^{\BI}_{\mathcal \I}$. Then
  there exists an instance
  $\f \in \BI({\mathcal \I})$ such that $\Imm (\f)=A$.
\end{Lem}

\begin{proof}
  Let $\mathcal{F}$ be the set of functions $f \in \BI({\mathcal \I})$ such that
  $\Imm(\f)\subseteq A$. By Cor.~\ref{unique} all functions in $\mathcal{F}$ have the same
  domain, therefore the expression $g=\bigoplus_{f\in\mathcal{F}} f$ is
  well-formed; by construction $\Imm(g) \subseteq A$.
  By definition of  $\PBI{\I}$ all functions $f\in\BI(\I)$ are such that $\Imm(\f)
  \subseteq A$ or $\Imm(g) \cap  A=\emptyset$, therefore all functions
  $f\in\BI(\I)$ whose image intersect $A$ are such that $\Imm(\f)\subseteq A$,
  which in turn implies that are also in $\mathcal{F}$.
  Hence  $\Imm (g)=A$, for otherwise there would be an element of $A$ not
  belonging to the image of any function in $\I$.
\end{proof}


\begin{Lem}
  \label{uni2}
  Let $\I$ be a set of instances over $B$, let $P^{\BI}_{\mathcal \I}$
  be the partition induced by ${\I}$ and let $A$ be a union of sets in
  $\PBI{\I}$, such that the inverse image of $A$ induces a weakly-connected
  subgraph of the schema, then $A$  is the image of
  an instance
  $\f \in \BI({\mathcal \I})$.
\end{Lem}

\begin{proof}
  Let  $A_1,\cdots,A_n$ be the sets of $\PBI{\I}$ whose union is $A$, and notice
  that, by Lemma~\ref{minimo_ottenibile}, it is possible to associate to each set
  $A_i$ the instance $\f_i \in \BI(\I)$ whose image is $A_i$, moreover for each such $\f_i$,
  $|Dom(\f_i)|=1$.
  For each vertex $x$ in the inverse image of $A$ we can construct the function
  $g_x$ as $\bigoplus_{Dom(f_i) = \{x\}}  f_i$.
  Then let $g= \bigotimes g_x$; it is immediate to note that $\Imm (g)=A$.
\end{proof}

An immediate consequence of Lemmata~\ref{f_1} and~\ref{uni2} is the following:
\begin{Cor}
  \label{canonical=obtainable}
  Let $\I$ be a set of instances over $B$ and let $f\in\I$. Then $f\in \BI(\I)$ if and
  only if $\Imm(f)$ is union of sets in $\PBI{\I}$ and the inverse image of
  $\Imm(f)$ induces a weakly-connected subgraph of the schema.
\end{Cor}

With Lemma~\ref{uni2} we have proved that all interesting unions of sets of the
partition  $P^{\BI}_{\mathcal \I}$ can be obtained with an expression of the
graph algebra where all operands are taken from  ${\mathcal \I}$, therefore
$P^{\BI}_{\mathcal \I}$ conveys all expressibility information. But
$P^{\BI}_{\mathcal \I}$ is defined on the set $\BI({\mathcal \I})$, we still
need to correlate the definition of canonical partition with that of
$P^{\BI}_{\mathcal \I}$.

\begin{Lem}
  \label{Pbi_accettabile}
  Let  ${\mathcal \I}$ be a set of instances. Then $P^{\BI}_{\mathcal \I}$
  is a valid partition of ${\mathcal \I}$.
\end{Lem}

\begin{proof}
  Let $A,B$ be two disjoint sets where $A\in \PBI{\I}$ and $B$ is union of sets in
  $\PBI{\I}$, we will prove that $\stronguniform{A}{B}{\I}$.
  First of all we will show that $A$ is 0-stable. Remember that, by definition of
  $\PBI{\I}$, for each $f\in \BI(\I)$ either $\Imm(f)\supseteq A$ or  $\Imm(f)\cap
  A=\emptyset$. Since $\BI(\I)$ contains $\I$, it is immediate to not that $A$ is
  0-stable. In the following let $a$ be the single-vertex inverse image of $A$.

  If the inverse image of $A\cup B$ does not induce a weakly-connected subgraph of
  the instance, then 0-stability of $A$ suffices to prove that
  $\stronguniform{A}{B}{\I}$, therefore assume that the inverse image of $A\cup B$
  induces a weakly-connected subgraph of the instance.
  Let us assume that
  $\stronguniform{A}{B}{\I}$ does not hold, then we will get a contradiction. Without loss of
  generality we can assume that $B$  is a minimum set for which
  $\stronguniform{A}{B}{\I}$ does not hold.
  The new assumption implies that $k$-stability does not hold for some $k$. It
  follows that
  there are two elements $x, y \in A$  and an element $z \in B$ such that
  $\AM_{A\cup B,{\mathcal \I}}(x,z)$ is not contained in $\AM_{A \cup B, {\mathcal \I}}(y,v)$, for every
  $v \in B$.
  Now, by Lemma~\ref{uni2}, there is an instance
  $\g \in \BI({\mathcal \I})$ such that $\Imm(\g) = A \cup B$.
  Let $\name(x)=x_1$ and $\name(z)=z_1$ and let $v \in B$ such that $\name(v)=z_1$.
  Let us consider the function
  $h=\bigoplus_{G_s\in\AM_{A \cup B, {\mathcal \I}}(x,z)} \left( \Pi_a (g|G_s) \right)$. By
  construction $x$ belongs to $\Imm(h)$, but $y$ does not; since $h\in \BI(\I)$ we
  have found a function in $\BI(\I)$ containing $x$ but not $y$, contradicting the
  assumption that $x R_I y$.
\end{proof}

\begin{Lem}
  \label{Pd=Pbi}
  Let  ${\mathcal \I}$ be a set of instances.
  Then $P_{\mathcal \I}^{\BI}={\mathcal
    C}_{\mathcal \I}$.
\end{Lem}

\begin{proof}
  By Lemma~\ref{Pbi_accettabile},  $P_{\I}^{\BI}$  must be
  a valid partition of ${\mathcal \I}$.
  By Lemma~\ref{minimo_ottenibile}, every set $A \in P_{\mathcal \I}^{\BI}$
  is obtained as the image of some instance $\f \in \BI({\mathcal \I})$.
  Clearly, since  $\f \in \BI({\mathcal \I})$,
  by Lemma~\ref{f_1} the image $A$
  of $\f$ is the union of sets of the canonical
  partition of ${\mathcal \I}$.
  Hence
  $\PBI{\I} = {\mathcal C}_{\I}$.
\end{proof}

\begin{Thm}
  \label{uguale1}
  Let  ${\mathcal \I}$ be a set of instances.
  Then an  instance $\g$ belongs to  $\BI (\I)$ if and only if $\PBI{\I}=\PBI{\I\cup \{ g\}}$.
\end{Thm}

\begin{proof}
  Clearly by Lemma~\ref{Pd=Pbi} it suffices to show that $\g \in \BI({\mathcal \I})$
  if and only if  $P_{\mathcal \I}^{\BI} = P_{{\mathcal \I}\cup \{ g\}}^{\BI}$,
  moreover it is immediate to note that, by construction of $\PBI{\I}$, if $\g \in \BI({\mathcal \I})$ then
  $P_{\mathcal \I}^{\BI} = P_{{\mathcal \I}\cup \{ g\}}^{\BI}$.
  Assume now that $P_{\mathcal \I}^{\BI} = P_{{\mathcal \I}\cup \{ g\}}^{\BI}$.
  By Lemma~\ref{Pbi_accettabile},  $P_{{\mathcal \I}\cup \{
    g\}}^{\BI}$ must be a valid partition.
  By Lemma~\ref{f_1}, and exploiting the assumption that $\PBI{\I}=\PBI{\I\cup \{
    g\}}$, for each $x \in \dom(\g)$,  $\g(x)$
  must be the union of some sets  in $P_{{\mathcal \I}\cup \{ g\}}^{\BI}$, that is
  $\Imm(\g(x))= \cup A_j$, for some sets $A_j \in P_{\mathcal \I}\cup \{ g\}^{\BI}$.
  But by Lemma~\ref{minimo_ottenibile}, for each set $A_j$ there is
  an instance $\f_j \in \BI({\mathcal \I})$ such that $A_j$ is the image of $\f_j$.
  Consequently, $\g(x) = \oplus \f_i$, and hence $\g = \otimes_{ x \in \dom(\g)}\g(x)$,
  which proves that $\g \in \BI({\mathcal \I})$ as required.
\end{proof}

Theorem~\ref{uguale1} and Lemma~\ref{Pd=Pbi} lead to our main result.

\begin{Cor}
  \label{cor:uguale1}
  Let  ${\mathcal \I}$ be a set of instances.
  Then an  instance $\g$ belongs to  $\BI (\I)$ if and only if
  $\mathcal{C}_{\I}=\mathcal{C}_{\I\cup \{ g\}}$.
\end{Cor}

\section{Conclusions}

We have introduced the idea that partitions of the domain set can be used for
characterizing the set of relations or graphs that can be extracted in a data
base in the relational or in a graph-based model.
By formally proving those expressiveness results we have effectively given a new
framework for the analysis of data base query languages.

The graph-based model presented here is not rich enough to be considered of
practical use, therefore it
would be interesting to use our framework for analyzing a more sofisticated
graph-based model.




\end{document}